\documentclass[journal,twoside,web]{ieeecolor}
\usepackage{generic}
\usepackage{cite}
\usepackage{amsmath,amssymb,amsfonts}
\usepackage{algorithm}
\usepackage{algpseudocode}
\newcommand\NoDo{\renewcommand\algorithmicdo{}}
\newcommand\NoThen{\renewcommand\algorithmicthen{}}
\usepackage{graphicx}
\usepackage{textcomp}
\usepackage{array}
\usepackage[caption=false]{subfig}
\usepackage{setspace}
\usepackage{epsfig}
\usepackage{cleveref}
\usepackage[caption=false]{subfig}
\usepackage{xcolor}
\DeclareMathOperator*{\argmax}{arg\,max}

\algblock{ParFor}{EndParFor}
\algnewcommand\algorithmicparfor{\textbf{parfor}}
\algnewcommand\algorithmicendparfor{\textbf{end\ parfor}}
\algrenewtext{ParFor}[1]{\algorithmicparfor\ #1}
\algrenewtext{EndParFor}{\algorithmicendparfor}
\newtheorem{theorem}{Theorem}
\newtheorem{lemma}{Lemma}

\newtheorem{remark}{Remark}
\newtheorem{definition}{Definition}

\usepackage{mathtools}

\makeatletter
  \def\my@tag@font{\normalsize}
  \def\maketag@@@#1{\hbox{\m@th\normalfont\my@tag@font#1}}
  \let\amsmath@eqref\eqref
  \renewcommand\eqref[1]{{\let\my@tag@font\relax\amsmath@eqref{#1}}}
\makeatother

\begin{document}
\title{Reducing Attack Opportunities Through Decentralized Event-Triggered Control}
\author{Paul Griffioen, \IEEEmembership{Student Member, IEEE}, Raffaele Romagnoli \IEEEmembership{Member, IEEE},\\Bruce H. Krogh \IEEEmembership{Fellow, IEEE}, and Bruno Sinopoli, \IEEEmembership{Fellow, IEEE}
\thanks{This work was supported in part by the National Science Foundation under Grant 1932530.
}
\thanks{P. Griffioen, R. Romagnoli, and B. H. Krogh are with the Department of Electrical and Computer Engineering, Carnegie Mellon University, Pittsburgh, PA, USA 15213. B. Sinopoli is with the Department of Electrical and Systems Engineering, Washington University in St. Louis, St. Louis, MO, USA 63130. Email: {\tt\small\{pgriffi1|rromagno|krogh\}@andrew.cmu.edu, bsinopoli@wustl.edu}}}

\maketitle

\begin{abstract}
Decentralized control systems are widely used in a number of situations and applications. In order for these systems to function properly and achieve their desired goals, information must be propagated between agents, which requires connecting to a network. To reduce opportunities for attacks that may be carried out through the network, we design an event-triggered mechanism for network connection and communication that minimizes the amount of time agents must be connected to the network, in turn decreasing communication costs. This mechanism is a function of only local information and ensures stability for the overall system in attack-free scenarios. Our approach distinguishes itself from current decentralized event-triggered control strategies by considering scenarios where agents are not always connected to the network to receive critical information from other agents and by considering scenarios where the communication graph is undirected and connected. An algorithm describing this network connection and communication protocol is provided, and our approach is illustrated via simulation.
\end{abstract}


\section{Introduction}
Cyber-physical systems, engineered systems which include sensing, communication, and control in physical spaces, are essential to secure and protect in today's society. Cyber-physical systems are ubiquitous in modern critical infrastructures including the smart grid, transportation systems, health care, sewage/water management, energy delivery, and manufacturing. These large scale, highly connected systems may be deployed in insecure public spaces and may contain heterogeneous components and devices, thus creating numerous attack surfaces. Consequently, these systems are attractive targets for adversaries, especially safety critical systems \cite{cardenas2008secure,slay2007lessons,case2016analysis,langner2011stuxnet}. Many of these systems are distributed over a wide area and are comprised of many different agents which interact with one another in a decentralized manner. It is important, therefore, to guarantee the safety and security of these decentralized control systems which rely heavily on network communication to achieve their goals.

In decentralized control systems, individual agents have access to differing amounts of information. In order to maintain the stability of the overall system and achieve a global objective, agents must occasionally communicate some subset of their local information with other agents, for instance in car platoons \cite{sun2003efficient,huang2016toward}. However, communicating with other agents requires connecting to the network, opening up the possibility for adversaries to corrupt information that is sent over the network and corrupt an agent's control software by using the network connection to inject malicious code \cite{szekeres2013sok}. These facts, in addition to the desire for minimizing communication costs, motivate the need for intermittent network connections as opposed to holding a constant network connection all the time. This is necessary in certain marine robotics applications where underwater vehicles must surface in order to share information with one another \cite{zolich2019survey}. Due to the fact that different sets of local information are available to each agent, the decision about when network connection and communication are necessary for a particular agent must be triggered locally. While intermittent network connections alone will not ensure resiliency against attacks, they reduce an adversary's window of opportunity for attack while also providing a framework in which a resilience strategy may be implemented.

Existing approaches to decentralized event-triggered control have mainly been concerned with minimizing communication costs, not with reducing opportunities for attacks that may be carried out through the network. Different approaches to event-triggered control are summarized well in \cite{heemels2012introduction}. A variety of decentralized event-triggered control mechanisms for linear, nonlinear, continuous time, and discrete time systems are presented in \cite{donkers2011networked,heemels2012periodic,donkers2011output,dolk2016output,heemels2013model,fu2020dynamic,fu2018decentralized} and provide conditions under which global asymptotic stability, global exponential stability, $\mathcal{L}_\infty$ gain performance, or $\mathcal{L}_p$ gain performance are achieved. All of these approaches assume that each agent is always connected to the network and is always available to receive any information that is sent to it, even though the agent might not always be broadcasting information to other agents. However, this assumption does not hold in contexts where safety and security are important since the attack window is minimized when each agent disconnects from the network for as long as possible.

In contrast to these previous approaches, we present a decentralized event-triggered network connection and communication protocol that does not assume that all agents are always connected to the network. The network connection and communication protocol ensures the stability of the overall system in attack-free scenarios when agents intermittently connect and disconnect from the network (as opposed to intermittently broadcasting information to other agents). This protocol uses trigger conditions based only on local information to determine when a particular agent must connect to the network to send and receive information from other agents. This trigger condition is designed to guarantee system stability by having an agent connect to the network when the magnitude of the state estimation error grows too large.

The main contributions of this article relative to our previous work in \cite{griffioen2020decentralized} and \cite{griffioen2021reducing} are as follows.
\begin{enumerate}
\item In contrast to \cite{griffioen2020decentralized}, we present a decentralized event-triggered network connection protocol that includes sensor measurements in the system model and does not require any reachability analysis to be implemented, which can be computationally costly.
\item In this article, only $2$ observer gain matrices need to be designed for each agent and only $2$ linear matrix inequalities (LMIs) need to be evaluated for the error dynamics. In contrast, $2^N-1$ observer gain matrices and $2^N-N$ LMIs are needed in \cite{griffioen2021reducing}, where $N$ is the number of agents, so this becomes computationally intractable with large numbers of agents.
\item This article is presented in continuous time with extensions to discrete time, unlike \cite{griffioen2020decentralized} and \cite{griffioen2021reducing} which are only presented in discrete time.
\item The restrictive assumption of a complete communication graph that is required in \cite{griffioen2020decentralized} and \cite{griffioen2021reducing} is relaxed in this article so that only the assumption of an undirected and connected communication graph is needed.
\end{enumerate}

The remainder of this article is organized as follows. Section II introduces the system model and estimation procedure that is used by each agent. Section III presents the triggering mechanism, network connection and communication protocol, and conditions under which stability of the overall system is achieved. Simulation results are presented in Section IV, and Section V concludes the article.

\section{Problem Formulation}

\subsection{Preliminaries}
A communication graph at time $t$ is denoted by $\mathcal{G}(t)=(\mathcal{N},\hat{\mathcal{A}}(t))$, where $\mathcal{N}=\{1,\cdots,N\}$ is a finite nonempty set of agents and $\hat{\mathcal{A}}(t)\in\mathbb{R}^{N\times N}$ represents the adjacency matrix at time $t$ whose $(i,j)^{th}$ element is given by $a_{ij}(t)\in\{0,1\}$. If $a_{ij}(t)=1$, agent $j$ is able to send information to agent $i$ at time $t$, and if $a_{ij}(t)=0$, agent $j$ is unable to send information to agent $i$ at time $t$. A communication graph is undirected at time $t$ if and only if $a_{ij}(t)=a_{ji}(t)$ $\forall i,j\in\mathcal{N}$. Self-connection is excluded so that $a_{ii}(t)=0$ $\forall i\in\mathcal{N}$, $\forall t$.

A path of length $\ell$ from agent $i$ to agent $j$ at time $t$ is a sequence $(i_0,\cdots,i_\ell)$ of agents such that $i_0=i$, $i_\ell=j$, and $a_{i_ki_{k-1}}(t)=1$ with distinct $i_k$'s. An undirected communication graph is connected at time $t$ if there is a path from $i$ to $j$ at time $t$ for two arbitrary distinct agents $i,j\in\mathcal{N}$.

The Laplacian of $\mathcal{G}(t)$ is the matrix $\mathcal{L}(t)$, where the $(i,j)^{th}$ element is given by $l_{ij}(t)$. The Laplacian is defined as $l_{ii}(t)\triangleq\sum_{j=1}^Na_{ij}(t)$ and $l_{ij}(t)\triangleq-a_{ij}(t)$, $i\neq j$. Note that $\mathcal{L}(t)1_N=0$ $\forall t$, where $1_N$ represents an $N\times1$ column vector comprised of all ones. By its construction, $\mathcal{L}(t)$ contains a zero eigenvalue with a corresponding eigenvector $1_N$, and all the other eigenvalues lie in the closed right-half complex plane. For undirected communication graphs at time $t$, both $\hat{\mathcal{A}}(t)$ and $\mathcal{L}(t)$ are symmetric and have real eigenvalues.

Let $\bar{\mathcal{L}}$ be the Laplacian for an undirected and connected communication graph. It is well known that $\bar{\mathcal{L}}$ has a simple zero eigenvalue, or in other words, all other eigenvalues are real positive. In this case, by Schur's lemma there exists a normal (real unitary) matrix $U\in\mathbb{R}^{N\times N}$ such that $U\bar{\mathcal{L}}U^T=\text{Diag}(0,\Lambda^+)$, where $\Lambda^+\in\mathbb{R}^{(N-1)\times(N-1)}$ is the diagonal matrix with entries corresponding to the positive eigenvalues of $\bar{\mathcal{L}}$. Given the properties of the Laplacian matrix, the first row of $U$ will be $\frac{1}{\sqrt{N}}1_N^T$. Then there is a matrix $S\in\mathbb{R}^{N\times(N-1)}$ such that
\begin{equation}
\medmuskip=3.56mu
\thinmuskip=3.56mu
\thickmuskip=3.56mu
U =
\begin{bmatrix}
\frac{1}{\sqrt{N}}1_N^T \\
S^T
\end{bmatrix}, ~
\begin{bmatrix}
\frac{1}{\sqrt{N}}1_N^T \\
S^T
\end{bmatrix} \bar{\mathcal{L}}
\begin{bmatrix}
\frac{1}{\sqrt{N}}1_N & S
\end{bmatrix} =
\begin{bmatrix}
0 & 0 \\
0 & \Lambda^+
\end{bmatrix}.
\end{equation}
In other words, there is a matrix $S\in\mathbb{R}^{N\times(N-1)}$ such that
\begin{equation}
1_N^TS=0, \quad S^TS=I_{N-1}, \quad S^T\bar{\mathcal{L}}S=\Lambda^+.
\end{equation}

\subsection{System Model}
We model the plant as a continuous time linear time invariant system composed of $N$ agents. The overall system dynamics are given by
\begin{align}
\label{StateDynamics1}
\dot{x}(t) &= Ax(t) + Bu(t) + w(t), \\
\label{Sensors1}
y(t) &= Cx(t) + v(t),
\end{align}
where $x(t)\in\mathbb{R}^n$ represents the state at time $t$, $u(t)\in\mathbb{R}^p$ denotes the control input, $y(t)\in\mathbb{R}^m$ represents the sensor measurements, and $w(t)\in\mathbb{R}^n$ and $v(t)\in\mathbb{R}^m$ are bounded disturbances which lie in the compact sets $W$ and $V$, respectively, given by
\begin{align}
W&\triangleq\left\{w(t)\middle|w(t)^TQw(t)\leq1\right\}, \\
V&\triangleq\left\{v(t)\middle|v(t)^TRv(t)\leq1\right\}.
\end{align}
We let $y_i(t)\in\mathbb{R}^{m_i}$ represent the sensor measurements that are locally available to agent $i$ so that $y(t)=\begin{bmatrix}y_1(t)^T&\cdots&y_N(t)^T\end{bmatrix}^T$, $m=\sum_{i=1}^Nm_i$, $C\triangleq\begin{bmatrix}C_1^T&\cdots&C_N^T\end{bmatrix}^T$, and $C_i\in\mathbb{R}^{m_i\times n}$. Similarly, we let $u_i(t)\in\mathbb{R}^{p_i}$ denote agent $i$'s control inputs so that $u(t)=\begin{bmatrix}u_1(t)^T&\cdots&u_N(t)^T\end{bmatrix}^T$, $p=\sum_{i=1}^Np_i$, $B\triangleq\begin{bmatrix}B_1&\cdots&B_N\end{bmatrix}$, and $B_i\in\mathbb{R}^{n\times p_i}$.

Each agent $i$ is able to directly access its own local sensor measurements $y_i(t)$ and control inputs $u_i(t)$ but must rely on communication from other agents to access any information locally available to other agents. We assume that the same underlying undirected and connected communication graph is always present, given by $\bar{\mathcal{G}}=(\mathcal{N},\bar{\mathcal{A}})$ with Laplacian matrix $\bar{\mathcal{L}}$ and with $\mathcal{N}_i$ denoting the set of agent $i$'s neighbors. When all agents are connected to the network, they are able to communicate with each another according to this underlying communication graph $\bar{\mathcal{G}}$. Consequently, if all agents are connected to the network at time $t$, $\mathcal{G}(t)=\bar{\mathcal{G}}$ so that $\hat{\mathcal{A}}(t)=\bar{\mathcal{A}}$ and $\mathcal{L}(t)=\bar{\mathcal{L}}$. If agent $i$ disconnects from the network at time $t$, it can no longer communicate with its neighbors, so $a_{ij}(t)=a_{ji}(t)=0$ $\forall j\in\mathcal{N}$. When agent $i$ connects to the network at time $t$, it is able to communicate with its neighbors who are also connected to the network at time $t$. Consequently, $a_{ij}(t)=a_{ji}(t)=1$ for all agents $j\in\mathcal{N}_i$ connected to the network at time $t$.

\subsection{State Estimation}
The control input for agent $i$ is given by
\begin{equation}
\label{ControlInput}
u_i(t) = K_i\hat{x}_i(t),
\end{equation}
where $\hat{x}_i(t)\in\mathbb{R}^n$ is agent $i$'s estimate of the overall state. Motivated by the distributed observer in \cite{kim2020decentralized}, this estimate is computed according to
\begin{equation}
\label{StateEstimateDynamics}
\begin{split}
\dot{\hat{x}}_i(t) &= A_{bk}\hat{x}_i(t) + \bar{L}_i(l_{ii}(t))(y_i(t)-C_i\hat{x}_i(t)) \\
&\quad+ \eta\sum_{j=1}^Na_{ij}(t)(\hat{x}_j(t)-\hat{x}_i(t)),
\end{split}
\end{equation}
where $A_{bk}\triangleq A+BK$, $K\triangleq\begin{bmatrix}K_1^T&\cdots&K_N^T\end{bmatrix}^T$, $\bar{L}_i(0)=\hat{L}_i$, $\bar{L}_i(\ell)=NL_i$ $\forall\ell\in\mathbb{Z}^+$, and $L\triangleq\begin{bmatrix}L_1&\cdots&L_N\end{bmatrix}$ with $L$ designed so that $A-LC$ is Hurwitz. Here $\hat{L}_i$ is agent $i$'s observer gain matrix, given by
\begin{equation}
\medmuskip=2.18mu
\thinmuskip=2.18mu
\thickmuskip=2.18mu
\label{ObserverDesign}
\hat{L}_i \triangleq T_i
\begin{bmatrix}
L_i^o \\
0
\end{bmatrix}, ~
T_i^{-1}AT_i =
\begin{bmatrix}
A_i^o & 0 \\
A_i^{21} & A_i^{\bar{o}}
\end{bmatrix}, ~
C_iT_i =
\begin{bmatrix}
C_i^o & 0
\end{bmatrix},
\end{equation}
where $L_i^o$ is designed so that $A_i^o-L_i^oC_i^o$ is Hurwitz, and $T_i$ is a similarity transformation matrix used to carry out the observability decomposition so that $(A_i^o,C_i^o)$ is observable. The observer presented in \eqref{StateEstimateDynamics} is simply a Luenberger observer that uses only local sensor measurements $y_i(t)$ when agent $i$ is not connected to the network. When agent $i$ is connected to the network, it broadcasts the current value of its state estimate to its neighbors, receives the current values of the state estimates of its neighbors who are connected to the network, and incorporates them into the coupling term of the observer, where $\eta\geq0$ is the coupling gain. By having agents intermittently connect and disconnect from the network, data is only sent over the network during intermittent intervals as opposed to being sent in a continuous stream over time. Note that agents do not need to have the same initial state estimate $\hat{x}_i(0)$.
\begin{remark}
Note that one of the benefits of this state estimation scheme compared to the one presented in \cite{griffioen2021reducing} is that each agent only uses $2$ observer gain matrices, $\bar{L}_i(0)=\hat{L}_i$ and $\bar{L}_i(\ell)=NL_i$ $\forall\ell\in\mathbb{Z}^+$. In contrast, $2^N-1$ observer gain matrices are needed in \cite{griffioen2021reducing}.
\end{remark}

\subsection{Problem Formulation}
Given the controller in \eqref{ControlInput}, the system dynamics in \eqref{StateDynamics1} can be written as
\begin{equation}
\label{StateDynamics2}
\dot{x}(t) = A_{bk}x(t) - Ee(t) + w(t),
\end{equation}
where $E\triangleq\begin{bmatrix}B_1K_1&\cdots&B_NK_N\end{bmatrix}$, $e(t)\triangleq\begin{bmatrix}e_1(t)^T&\cdots&e_N(t)^T\end{bmatrix}^T$, and $e_i(t)\triangleq x(t)-\hat{x}_i(t)$ so that $e(t)\in\mathbb{R}^{Nn}$ and $e_i(t)\in\mathbb{R}^n$. Here $e_i(t)$ represents the error between the overall state and agent $i$'s estimate of the overall state. We next introduce a network connection protocol which decides when it is necessary for each agent to connect to the network and communicate with other agents to ensure the stability of the overall system.

\section{Network Connection Protocol}

\subsection{Quadratic Boundedness}
In order to maintain the stability and safety of the overall system, each agent $i$ must occasionally connect to the network and share information with its neighboring agents. We would like to design a network connection protocol that ensures the stability of the overall system by properly coordinating communication between different agents while also minimizing the number of times each agent connects to the network. The mechanism that triggers this network connection is only able to use locally available information, which includes information that has been received from other agents through past communications. The stability which we design the network connection protocol to achieve is quadratic $\gamma$-boundedness which is described in Definition \ref{QuadraticBoundedness}, Definition \ref{RobustPositiveInvariance}, and Lemma \ref{EquivalenceLemma}. These definitions and lemmas have been uniquely modified and adapted from \cite{brockman1998quadratic} by adding the parameter $\gamma$ to not only specify \textit{when} the Lyapunov function decreases or increases but also \textit{how fast} it does so.
\begin{definition}[\cite{brockman1998quadratic}]
\label{QuadraticBoundedness}
Let $z(t)$ represent a state vector, let $d_1(t)$ and $d_2(t)$ represent disturbance vectors, and let $D_1$ and $D_2$ be compact sets. A system of the form
\begin{equation}
\label{QuadraticBoundSystem}
\medmuskip=0.84mu
\thinmuskip=0.84mu
\thickmuskip=0.84mu
\dot{z}(t) = \mathcal{A}z(t) + \mathcal{B}_1d_1(t) + \mathcal{B}_2d_2(t), ~ d_1(t)\in D_1,~d_2(t)\in D_2
\end{equation}
is quadratically $\gamma$-bounded with $\gamma\in\mathbb{R}$ and symmetric positive definite Lyapunov matrix $\mathcal{P}$ if and only if $\forall d_1(t)\in D_1$ and $\forall d_2(t)\in D_2$,
\begin{equation}
\medmuskip=3.73mu
\thinmuskip=3.73mu
\thickmuskip=3.73mu
z(t)^T\mathcal{P}z(t)\geq1\implies \frac{d}{dt}\left(z(t)^T\mathcal{P}z(t)\right)<\gamma z(t)^T\mathcal{P}z(t).
\end{equation}
\end{definition}
Note that when $\gamma\leq0$, $\gamma$ specifies how quickly the Lyapunov function decreases over time, and when $\gamma>0$, $\gamma$ sets an upper bound on how quickly the Lyapunov function can increase over time.
\begin{definition}[\cite{brockman1998quadratic}]
\label{RobustPositiveInvariance}
The set $Z$ is a robustly positively invariant set for \eqref{QuadraticBoundSystem} if and only if $z(0)\in Z$ implies that $z(t)\in Z$ $\forall t\geq0$, $\forall d_1(t)\in D_1$, and $\forall d_2(t)\in D_2$.
\end{definition}
\begin{lemma}[\cite{brockman1998quadratic}]
\label{EquivalenceLemma}
The following two statements are equivalent:
\begin{enumerate}
\item System \eqref{QuadraticBoundSystem} is quadratically $\gamma$-bounded with $\gamma\leq0$ and symmetric positive definite Lyapunov matrix $\mathcal{P}$.
\item The set $Z\triangleq\{z(t)^T\mathcal{P}z(t)\leq1,~\mathcal{P}\succ0\}$ is a robustly positively invariant set for \eqref{QuadraticBoundSystem}.
\end{enumerate}
\end{lemma}

Given these definitions, Lemma \ref{QuadraticBoundLemma} provides a sufficient condition for evaluating the quadratic $\gamma$-boundedness of \eqref{QuadraticBoundSystem}.
\begin{lemma}
\label{QuadraticBoundLemma}
Let $D_1\triangleq\{d_1(t)|d_1(t)^T\mathcal{D}_1d_1(t)\leq1,~\mathcal{D}_1\succ0\}$ and $D_2\triangleq\{d_2(t)|d_2(t)^T\mathcal{D}_2d_2(t)\leq1,~\mathcal{D}_2\succ0\}$ for the system in \eqref{QuadraticBoundSystem}. If $\exists\alpha\geq0$ such that
\begin{equation}
\label{QuadraticLMI}
\begin{bmatrix}
(\gamma-2\alpha)\mathcal{P}-\mathcal{A}^T\mathcal{P}-\mathcal{P}\mathcal{A} & -\mathcal{P}\mathcal{B}_1 & -\mathcal{P}\mathcal{B}_2 \\
-\mathcal{B}_1^T\mathcal{P} & \alpha\mathcal{D}_1 & 0 \\
-\mathcal{B}_2^T\mathcal{P} & 0 & \alpha\mathcal{D}_2
\end{bmatrix}\succ0,
\end{equation}
then the system in \eqref{QuadraticBoundSystem} is quadratically $\gamma$-bounded with symmetric positive definite Lyapunov matrix $\mathcal{P}$.
\end{lemma}
\begin{proof}
By using the S-procedure \cite{boyd1994linear}, \eqref{QuadraticLMI} is equivalent to
\begin{equation}
\medmuskip=-0.65mu
\thinmuskip=-0.65mu
\thickmuskip=-0.65mu
\begin{split}
&\begin{bmatrix}
z(t) \\ d_1(t) \\ d_2(t)
\end{bmatrix}^T
\begin{bmatrix}
-2\mathcal{P} & 0 & 0 \\
0 & \mathcal{D}_1 & 0 \\
0 & 0 & \mathcal{D}_2
\end{bmatrix}
\begin{bmatrix}
z(t) \\ d_1(t) \\ d_2(t)
\end{bmatrix}\leq0 ~\implies \\
&\begin{bmatrix}
z(t) \\ d_1(t) \\ d_2(t)
\end{bmatrix}^T
\begin{bmatrix}
\mathcal{A}^T\mathcal{P}+\mathcal{P}\mathcal{A}-\gamma\mathcal{P} & \mathcal{P}\mathcal{B}_1 & \mathcal{P}\mathcal{B}_2 \\
\mathcal{B}_1^T\mathcal{P} & 0 & 0 \\
\mathcal{B}_2^T\mathcal{P} & 0 & 0
\end{bmatrix}
\begin{bmatrix}
z(t) \\ d_1(t) \\ d_2(t)
\end{bmatrix}<0,
\end{split}
\end{equation}
which in turn is equivalent to
\begin{equation}
\label{Implication1}
\medmuskip=3.57mu
\thinmuskip=3.57mu
\thickmuskip=3.57mu
\begin{split}
&-2z(t)^T\mathcal{P}z(t) + d_1(t)^T\mathcal{D}_1d_1(t) + d_2(t)^T\mathcal{D}_2d_2(t) \leq 0 \\
&\hspace{2.27cm} \implies \frac{d}{dt}\left(z(t)^T\mathcal{P}z(t)\right) < \gamma z(t)^T\mathcal{P}z(t).
\end{split}
\end{equation}
Note that
\begin{equation}
\label{Implication2}
\medmuskip=2.94mu
\thinmuskip=2.94mu
\thickmuskip=2.94mu
\begin{split}
&\left\{\begin{matrix}
z(t)^T\mathcal{P}z(t)\geq1\\
d_1(t)^T\mathcal{D}_1d_1(t)\leq1\\
d_2(t)^T\mathcal{D}_2d_2(t)\leq1
\end{matrix}\right\} \implies \\
&-2z(t)^T\mathcal{P}z(t) + d_1(t)^T\mathcal{D}_1d_1(t) + d_2(t)^T\mathcal{D}_2d_2(t) \leq 0.
\end{split}
\end{equation}
Taking \eqref{Implication1} and \eqref{Implication2} in conjunction with one another yields that $\forall d_1(t)\in D_1$ and $\forall d_2(t)\in D_2$,
\begin{equation}
\medmuskip=3.73mu
\thinmuskip=3.73mu
\thickmuskip=3.73mu
z(t)^T\mathcal{P}z(t) \geq 1 \implies \frac{d}{dt}\left(z(t)^T\mathcal{P}z(t)\right) < \gamma z(t)^T\mathcal{P}z(t),
\end{equation}
implying that the system in \eqref{QuadraticBoundSystem} is quadratically $\gamma$-bounded with symmetric positive definite Lyapunov matrix $\mathcal{P}$.
\end{proof}

Lemma \ref{StateConvergenceLemma} provides a sufficient condition under which the overall system is quadratically $0$-bounded and the state remains in the robust positive invariant set $\mathcal{E}_x$ given by
\begin{equation}
\label{StateInvariantSet}
\mathcal{E}_x\triangleq\left\{x(t)\middle|x(t)^TPx(t)\leq1,~P\succ0\right\}.
\end{equation}
\begin{lemma}
\label{StateConvergenceLemma}
If $\exists\alpha_1\geq0$ such that
\begin{equation}
\label{StateLMI}
\begin{bmatrix}
-2\alpha_1P-A_{bk}^TP-PA_{bk} & PE & -P \\
E^TP & \alpha_1\bar{P} & 0 \\
-P & 0 & \alpha_1Q
\end{bmatrix}\succ0,
\end{equation}
then the system in \eqref{StateDynamics2} is quadratically $0$-bounded with symmetric positive definite Lyapunov matrix $P$ when $e(t)\in\mathcal{E}_e$, where $\mathcal{E}_e$ is given by
\begin{equation}
\label{ErrorInvariantSet}
\mathcal{E}_e\triangleq\left\{e(t)\middle|e(t)^T\bar{P}e(t)\leq1,~\bar{P}\succ0\right\}.
\end{equation}
Furthermore, if $x(0)\in\mathcal{E}_x$, then $x(t)\in\mathcal{E}_x$ $\forall t\geq0$.
\end{lemma}
\begin{proof}
Applying Lemma \ref{QuadraticBoundLemma} to the system in \eqref{StateDynamics2} implies that if \eqref{StateLMI} is satisfied, then the system in \eqref{StateDynamics2} is quadratically $0$-bounded. Lemma \ref{EquivalenceLemma} and Definition \ref{RobustPositiveInvariance} imply that if $x(0)\in\mathcal{E}_x$, then $x(t)\in\mathcal{E}_x$ $\forall t\geq0$.
\end{proof}

According to \eqref{StateDynamics1}, \eqref{Sensors1}, \eqref{ControlInput}, and \eqref{StateEstimateDynamics}, the error dynamics for agent $i$ are given by
\begin{equation}
\begin{split}
\dot{e}_i(t) &= (A_{bk}-\bar{L}_i(l_{ii}(t))C_i)e_i(t) + w(t) \\
&\quad - \bar{L}_i(l_{ii}(t))v_i(t) - \sum_{j=1}^N\left(B_jK_j+\eta l_{ij}(t)\right)e_j(t).
\end{split}
\end{equation}
The error dynamics for the overall system are then given by
\begin{equation}
\medmuskip=-0.9mu
\thinmuskip=-0.9mu
\thickmuskip=-0.9mu
\label{ErrorDynamics}
\dot{e}(t) = \underbrace{\left(F(\mathcal{L}(t))-\eta(\mathcal{L}(t)\otimes I_n)\right)}_{A_e(\mathcal{L}(t))}e(t) + \mathcal{I}w(t) - J(\mathcal{L}(t))v(t),
\end{equation}
where $\mathcal{I}\triangleq\begin{bmatrix}I_n&\cdots&I_n\end{bmatrix}^T$, $J(\mathcal{L}(t))\triangleq\text{BlkDiag}(\bar{L}_1(l_{11}(t)),\cdots,\bar{L}_N(l_{NN}(t)))$, and $F(\mathcal{L}(t))\triangleq$
\begin{equation*}
\footnotesize
\medmuskip=-0.37mu
\thinmuskip=-0.37mu
\thickmuskip=-0.37mu
\begin{bmatrix}
A_{bk}-\bar{L}_1(l_{11}(t))C_1-B_1K_1 & \cdots & -B_NK_N \\
\vdots & \ddots & \vdots \\
-B_1K_1 & \cdots & A_{bk}-\bar{L}_N(l_{NN}(t))C_N-B_NK_N
\end{bmatrix}.
\end{equation*}

Lemma \ref{ErrorInvarianceLemma} provides a sufficient condition under which the error is quadratically $\gamma(\mathcal{L}(t))$-bounded, where $\gamma$ is a function of the specific configuration of the communication graph at time $t$.
\begin{lemma}
\label{ErrorInvarianceLemma}
If $\exists\alpha_2\geq0$ such that
\begin{equation}
\footnotesize
\medmuskip=-0.67mu
\thinmuskip=-0.67mu
\thickmuskip=-0.67mu
\label{ErrorLMI}
\begin{bmatrix}
(\gamma(\mathcal{L}(t))-2\alpha_2)\bar{P}-A_e(\mathcal{L}(t))^T\bar{P}-\bar{P}A_e(\mathcal{L}(t)) & -\bar{P}\mathcal{I} & \bar{P}J(\mathcal{L}(t)) \\
-\mathcal{I}^T\bar{P} & \alpha_2Q & 0 \\
J(\mathcal{L}(t))^T\bar{P} & 0 & \alpha_2R
\end{bmatrix}\succ0,
\end{equation}
then the error in \eqref{ErrorDynamics} is quadratically $\gamma(\mathcal{L}(t))$-bounded with symmetric positive definite Lyapunov matrix $\bar{P}$.
\end{lemma}
\begin{proof}
Applying Lemma \ref{QuadraticBoundLemma} to the system in \eqref{ErrorDynamics} implies that if \eqref{ErrorLMI} is satisfied, then the error in \eqref{ErrorDynamics} is quadratically $\gamma(\mathcal{L}(t))$-bounded.
\end{proof}

As shown in \cite{kim2020decentralized}, in the case where all agents are connected to the network, we can decompose the error $e(t)$ into its average $\bar{e}(t)\triangleq\frac{1}{N}\mathcal{I}^Te(t)$ and the rest $\tilde{e}(t)\triangleq(S^T\otimes I_n)e(t)$ so that
\begin{equation}
e(t) =
\underbrace{\begin{bmatrix}
\mathcal{I} & S\otimes I_n
\end{bmatrix}}_{\bar{E}}
\underbrace{\begin{bmatrix}
\bar{e}(t) \\
\tilde{e}(t)
\end{bmatrix}}_{\hat{e}(t)}.
\end{equation}
When all the agents connect to the network, the error dynamics for the overall system are given by
\begin{equation}
\medmuskip=-1.65mu
\thinmuskip=-1.65mu
\thickmuskip=-1.65mu
\label{ErrorDynamics2}
\begin{split}
\begin{bmatrix}
\dot{\bar{e}}(t) \\
\dot{\tilde{e}}(t)
\end{bmatrix} &=
\underbrace{\begin{bmatrix}
A-LC & H(S\otimes I_n) \\
(S^T\otimes I_n)\bar{F}\mathcal{I} & (S^T\otimes I_n)\bar{F}(S\otimes I_n)-\eta(\Lambda^+\otimes I_n)
\end{bmatrix}}_{\bar{A}_e}
\begin{bmatrix}
\bar{e}(t) \\
\tilde{e}(t)
\end{bmatrix} \\
&\quad+\underbrace{\begin{bmatrix}
I_n \\
(S^T\otimes I_n)\mathcal{I}
\end{bmatrix}}_{\bar{W}} w(t) -
\underbrace{\begin{bmatrix}
L \\
(S^T\otimes I_n)\bar{J}
\end{bmatrix}}_{\bar{V}} v(t),
\end{split}
\end{equation}
where $\bar{J}\triangleq N\text{BlkDiag}(L_1,\cdots,L_N)$,
\begin{gather*}
\medmuskip=-0.25mu
\thinmuskip=-0.25mu
\thickmuskip=-0.25mu
H \triangleq
\begin{bmatrix}
\frac{1}{N}A_{bk}-B_1K_1-L_1C_1 & \cdots & \frac{1}{N}A_{bk}-B_NK_N-L_NC_N
\end{bmatrix}, \\
\medmuskip=-1.54mu
\thinmuskip=-1.54mu
\thickmuskip=-1.54mu
\bar{F} \triangleq
\begin{bmatrix}
A_{bk}-NL_1C_1-B_1K_1 & \cdots & -B_NK_N \\
\vdots & \ddots & \vdots \\
-B_1K_1 & \cdots & A_{bk}-NL_NC_N-B_NK_N
\end{bmatrix}.
\end{gather*}
Since $\Lambda^+$ is positive definite, the coupling gain $\eta$ can be designed with a sufficiently large value to ensure that the lower right block of $\bar{A}_e$ is negative definite. This ensures that the diagonal blocks of $\bar{A}_e$ only contain eigenvalues that lie in the left-half plane since $L$ has been designed to make $A-LC$ Hurwitz.

Lemma \ref{ErrorInvarianceLemma2} provides a sufficient condition under which the error is quadratically $\gamma(\bar{\mathcal{L}})$-bounded when all the agents connect to the network.
\begin{lemma}
\label{ErrorInvarianceLemma2}
If $\exists\alpha_3\geq0$ such that
\begin{equation}
\scriptsize
\medmuskip=-3.18mu
\thinmuskip=-3.18mu
\thickmuskip=-3.18mu
\label{ErrorLMI2}
\begin{bmatrix}
(\gamma(\bar{\mathcal{L}})-2\alpha_3)\bar{E}^T\bar{P}\bar{E}-\bar{A}_e^T\bar{E}^T\bar{P}\bar{E}-\bar{E}^T\bar{P}\bar{E}\bar{A}_e & -\bar{E}^T\bar{P}\bar{E}\bar{W} & \bar{E}^T\bar{P}\bar{E}\bar{V} \\
-\bar{W}^T\bar{E}^T\bar{P}\bar{E} & \alpha_3Q & 0 \\
\bar{V}^T\bar{E}^T\bar{P}\bar{E} & 0 & \alpha_3R
\end{bmatrix}\succ0,
\end{equation}
then the error in \eqref{ErrorDynamics} is quadratically $\gamma(\bar{\mathcal{L}})$-bounded with symmetric positive definite Lyapunov matrix $\bar{P}$ when all the agents are connected to the network.
\end{lemma}
\begin{proof}
Applying Lemma \ref{QuadraticBoundLemma} to the system in \eqref{ErrorDynamics2} implies that if \eqref{ErrorLMI2} is satisfied, then the error $\hat{e}(t)$ in \eqref{ErrorDynamics2} is quadratically $\gamma(\bar{\mathcal{L}})$-bounded with symmetric positive definite Lyapunov matrix $\bar{E}^T\bar{P}\bar{E}$. Since $\hat{e}(t)^T\bar{E}^T\bar{P}\bar{E}\hat{e}(t)=e(t)^T\bar{P}e(t)$, this is equivalent to the error in \eqref{ErrorDynamics} being quadratically $\gamma(\bar{\mathcal{L}})$-bounded with symmetric positive definite Lyapunov matrix $\bar{P}$ when all the agents are connected to the network.
\end{proof}

\subsection{Stability Conditions}
We want to ensure that $x(t)$ converges to $\mathcal{E}_x$ by creating a network connection protocol which guarantees that when the Lyapunov function $V(x(t))\triangleq x(t)^TPx(t)$ is greater than or equal to $1$, it decreases over time and converges to the robust positive invariant set $\mathcal{E}_x$. The invariance of $\mathcal{E}_x$ is shown in Lemma \ref{EquivalenceLemma} to be equivalent to quadratic $0$-boundedness. Consequently, the network connection protocol should guarantee that when $V(x(t))\geq1$, $\dot{V}(x(t))<0$ $\forall t$. The following theorem, motivated by \cite{donkers2011networked} and \cite{heemels2012periodic}, sets forth sufficient conditions under which $\dot{V}(x(t))<0$ $\forall t$.
\begin{theorem}
\label{QuadraticBoundednessTheorem}
If the network connection protocol ensures that for some $i\in\{1,\cdots,N\}$,
\begin{equation}
\label{SatisfactionCriteria}
-y_i(t)^TY_iy_i(t)+e(t)^T\bar{P}e(t)+w(t)^TQw(t)+v(t)^TRv(t)<0,
\end{equation}
where $Y_i\succ0$, and if $\forall i\in\{1,\cdots,N\}$,
\begin{equation}
\medmuskip=2.46mu
\thinmuskip=2.46mu
\thickmuskip=2.46mu
\label{LMICondition}
\begin{bmatrix}
-A_{bk}^TP-PA_{bk}-C_i^TY_iC_i & PE & -P & -C_i^TY_i\Gamma_i \\
E^TP & \bar{P} & 0 & 0 \\
-P & 0 & Q & 0 \\
-\Gamma_i^TY_iC_i & 0 & 0 & R-\Gamma_i^TY_i\Gamma_i
\end{bmatrix} \succeq 0,
\end{equation}
where $\Gamma_i\triangleq\begin{bmatrix}0_{m_i\times\sum_{j=1}^{i-1}m_j}&I_{m_i}&0_{m_i\times\sum_{j=i+1}^Nm_j}\end{bmatrix}$, then $\dot{V}(x(t))<0$ $\forall t$ for the system in \eqref{StateDynamics2}.
\end{theorem}
\begin{proof}
The condition in \eqref{LMICondition} is equivalent to
\begin{equation}
\begin{split}
&w(t)^TQw(t)-(A_{bk}x(t)-Ee(t)+w(t))^TPx(t)+\\
&v(t)^TRv(t)-x(t)^TP(A_{bk}x(t)-Ee(t)+w(t))\\
&\hspace{2.77cm}+e(t)^T\bar{P}e(t)-y_i(t)^TY_iy_i(t)\geq0,
\end{split}
\end{equation}
which is equivalent to
\begin{equation}
\label{LyapunovBound}
\begin{split}
\dot{V}(x(t))&\leq-y_i(t)^TY_iy_i(t)+e(t)^T\bar{P}e(t)\\
&\quad+w(t)^TQw(t)+v(t)^TRv(t).
\end{split}
\end{equation}
Taking \eqref{LyapunovBound} in conjunction with the condition in \eqref{SatisfactionCriteria} ensures that $\dot{V}(x(t))<0$ $\forall t$ for the system in \eqref{StateDynamics2} when \eqref{LMICondition} is satisfied $\forall i\in\{1,\cdots,N\}$ and \eqref{SatisfactionCriteria} is satisfied for some $i\in\{1,\cdots,N\}$.
\end{proof}

Note that $P$ determines the size of the invariant set $\mathcal{E}_x$ to which the state converges, $\bar{P}$ determines the size of the invariant set $\mathcal{E}_e$ to which the error converges, and $Y_i$ will have a direct impact on the frequency at which agent $i$ connects to the network as will be seen in \eqref{TriggerCondition}. Maximizing $\log\det{P}$ is proportional to minimizing the volume of $\mathcal{E}_x$, compressing the size of the invariant set to which the state converges. Maximizing $\log\det{\bar{P}}$ is proportional to minimizing the volume of $\mathcal{E}_e$, compressing the size of the invariant set to which the error converges. Maximizing $\log\det{Y_i}$ is proportional to maximizing $y_i(t)^TY_iy_i(t)$ $\forall y_i(t)$, which minimizes the number of times agent $i$ connects to the network as will be seen in \eqref{TriggerCondition}. Consequently, the desired values for $P$, $\bar{P}$, and $Y_i$ are obtained according to the following optimization problem
\begin{equation}
\medmuskip=1.62mu
\thinmuskip=1.62mu
\thickmuskip=1.62mu
\label{OptimizationProblem}
\begin{split}
&\argmax_{\alpha_1,\alpha_3,P,\bar{P},Y_1,\cdots,Y_N}\omega_x\log\det{P}+\omega_e\log\det{\bar{P}}+\sum_{i=1}^N\omega_i\log\det{Y_i}\\
&\text{s.t.}~\gamma(\bar{\mathcal{L}})\leq0,~~\alpha_1,\alpha_3\geq0,~~P\succ0,~~\bar{P}\succ0,\\
&\hspace{0.55cm}Y_i\succ0~~\forall i\in\{1,\cdots,N\},~~\text{\eqref{StateLMI}, \eqref{ErrorLMI2}, and \eqref{LMICondition} are satisfied,}
\end{split}
\end{equation}
where $\omega_x$, $\omega_e$, and $\omega_i$, $i\in\{1,\cdots,N\}$ are nonnegative constants chosen by the designer to weight the importance of minimizing $\mathcal{E}_x$, $\mathcal{E}_e$, and the communication frequency of agent $i$, respectively. The constraint $\gamma(\bar{\mathcal{L}})\leq0$ is included in this optimization problem because we would like the error to decrease when all agents are connected to the network as will be seen in Theorem \ref{TriggerTheorem}. Because this optimization problem is not convex, a suboptimal solution may be obtained by restricting the possible values of $\alpha_1$ and $\alpha_3$ to a finite set and carrying out the optimization problem in \eqref{OptimizationProblem} over that finite set, since \eqref{OptimizationProblem} is convex for set values of $\alpha_1$ and $\alpha_3$.

\subsection{Triggering Conditions}
The following theorem leverages the results of Lemmas \ref{StateConvergenceLemma}-\ref{ErrorInvarianceLemma2} and Theorem \ref{QuadraticBoundednessTheorem}, providing a set of network connection triggering conditions for each agent based on \eqref{SatisfactionCriteria} to ensure that $x(t)$ converges to $\mathcal{E}_x$.
\begin{theorem}
\label{TriggerTheorem}
Let $\{\bar{t}_k\}$, $k\in\mathbb{Z}^+$ represent the sequence of time instants where the configuration of the communication graph changes due to agents connecting or disconnecting from the network. We let $\bar{t}_0\triangleq0$. If $\exists\alpha_1,\alpha_3\geq0$ such that \eqref{StateLMI} and \eqref{ErrorLMI2} are satisfied with $\gamma(\bar{\mathcal{L}})\leq0$, if $e(0)\in\mathcal{E}_e$, if \eqref{LMICondition} is satisfied $\forall i\in\{1,\cdots,N\}$, and if agent $i$ connects to the network when
\begin{equation}
\medmuskip=2.27mu
\thinmuskip=2.27mu
\thickmuskip=2.27mu
\label{TriggerCondition}
y_i(t)^TY_iy_i(t) \leq 2 + \max\left(1,e^{\gamma(\bar{t}_k)(t-\bar{t}_k)+\sum_{j=0}^{k-1}\gamma(\bar{t}_j)(\bar{t}_{j+1}-\bar{t}_j)}\right),
\end{equation}
for all $t\in[\bar{t}_k,\bar{t}_{k+1})$, $\forall k\in\mathbb{Z}^+$, where
\begin{equation}
\label{GammaDefinition}
\begin{split}
\gamma(\bar{t}_k) &\triangleq \min_{\alpha_2,\gamma(\mathcal{L}(t))}\gamma(\mathcal{L}(t)) \\
&\quad~\text{s.t.}~\alpha_2\geq0,~\text{\eqref{ErrorLMI} is satisfied}~\forall t\in[\bar{t}_k,\bar{t}_{k+1}),
\end{split}
\end{equation}
then $x(t)$ will converge to $\mathcal{E}_x$ as long as agent $i$ stays connected to the network while
\begin{equation}
\label{StayConnected}
\begin{split}
\int_0^t\bigg(&2 + \max\left(1,e^{\gamma(\bar{t}_k)(\varphi-\bar{t}_k)+\sum_{j=0}^{k-1}\gamma(\bar{t}_j)(\bar{t}_{j+1}-\bar{t}_j)}\right) \\
&- y_i(\varphi)^TY_iy_i(\varphi)\bigg)d\varphi > f(t)
\end{split}
\end{equation}
for all $t\in[\bar{t}_k,\bar{t}_{k+1})$, $\forall k\in\mathbb{Z}^+$ when both agent $i$ and all its neighbors are connected to the network, where $f(t):\mathbb{R}\to\mathbb{R}$ is any strictly decreasing function such that $f(t)\leq0$ $\forall t\geq0$.
\end{theorem}
\begin{proof}
By applying Lemma \ref{QuadraticBoundLemma} and Definition \ref{QuadraticBoundedness} to the system in \eqref{ErrorDynamics}, the definition of $\gamma(\bar{t}_k)$ in \eqref{GammaDefinition} implies that $\forall w(t)\in W$ and $\forall v(t)\in V$,
\begin{equation}
\medmuskip=1.18mu
\thinmuskip=1.18mu
\thickmuskip=1.18mu
V_e(e(t)) \geq 1 \implies \dot{V}_e(e(t)) < \gamma(\bar{t}_k)V_e(e(t)) ~ \forall t\in[\bar{t}_k,\bar{t}_{k+1}),
\end{equation}
where $V_e(e(t))\triangleq e(t)^T\bar{P}e(t)$. This is equivalent to
\begin{equation}
\begin{split}
&V_e(e(t)) \geq 1 ~ \forall t\in[\bar{t}_k,\bar{t}_{k+1}) \implies \\
&\hspace{2.18cm}V_e(e(\bar{t}_{k+1})) < e^{\gamma(\bar{t}_k)(\bar{t}_{k+1}-\bar{t}_k)}V_e(e(\bar{t}_k)).
\end{split}
\end{equation}
Consequently, $\forall t\in[\bar{t}_k,\bar{t}_{k+1})$,
\begin{equation} 
\medmuskip=2.49mu
\thinmuskip=2.49mu
\thickmuskip=2.49mu
V_e(e(t)) < \max\left(1,e^{\gamma(\bar{t}_k)(t-\bar{t}_k)+\sum_{j=0}^{k-1}\gamma(\bar{t}_j)(\bar{t}_{j+1}-\bar{t}_j)}V_e(e(0))\right).
\end{equation}
If $e(0)\in\mathcal{E}_e$, then $\max_{e(0)}V_e(e(0))=1$, implying that $\forall t\in[\bar{t}_k,\bar{t}_{k+1})$,
\begin{equation}
\label{ErrorBound}
V_e(e(t)) < \max\left(1,e^{\gamma(\bar{t}_k)(t-\bar{t}_k)+\sum_{j=0}^{k-1}\gamma(\bar{t}_j)(\bar{t}_{j+1}-\bar{t}_j)}\right).
\end{equation}
The trigger condition in \eqref{TriggerCondition} can be written equivalently as
\begin{equation}
\medmuskip=0.69mu
\thinmuskip=0.69mu
\thickmuskip=0.69mu
2 - y_i(t)Y_iy_i(t) + \max\left(1,e^{\gamma(\bar{t}_k)(t-\bar{t}_k)+\sum_{j=0}^{k-1}\gamma(\bar{t}_j)(\bar{t}_{j+1}-\bar{t}_j)}\right) \geq 0.
\end{equation}
Note that $\forall w(t)\in W$ and $\forall v(t)\in V$,
\begin{equation}
\medmuskip=1.12mu
\thinmuskip=1.12mu
\thickmuskip=1.12mu
\begin{split}
& 2 - y_i(t)^TY_iy_i(t) + \max\left(1,e^{\gamma(\bar{t}_k)(t-\bar{t}_k)+\sum_{j=0}^{k-1}\gamma(\bar{t}_j)(\bar{t}_{j+1}-\bar{t}_j)}\right) > \\
& w(t)^TQw(t) + v(t)^TRv(t) - y_i(t)^TY_iy_i(t) + V_e(e(t))
\end{split}
\end{equation}
for all $t\in[\bar{t}_k,\bar{t}_{k+1})$, $\forall k\in\mathbb{Z}^+$ according to \eqref{ErrorBound}. Consequently, \eqref{TriggerCondition} functions as an upper bound on the condition in \eqref{SatisfactionCriteria} so that whenever \eqref{SatisfactionCriteria} is not satisfied, the condition in \eqref{TriggerCondition} will be triggered.

If \eqref{LMICondition} is satisfied $\forall i\in\{1,\cdots,N\}$, Theorem \ref{QuadraticBoundednessTheorem} states that $\dot{V}(x(t))<0$ $\forall t$ when \eqref{TriggerCondition} is not triggered for at least one agent, implying that $x(t)$ converges to $\mathcal{E}_x$ when \eqref{TriggerCondition} is not triggered for some $i\in\{1,\cdots,N\}$. In the case where \eqref{TriggerCondition} is triggered $\forall i\in\{1,\cdots,N\}$, all agents will be connected to the network, and either $e(t)\in\mathcal{E}_e$ or $e(t)\not\in\mathcal{E}_e$. If $e(t)\in\mathcal{E}_e$, then Lemma \ref{StateConvergenceLemma} states that if $\exists\alpha_1\geq0$ such that \eqref{StateLMI} is satisfied, then $x(t)$ will converge to $\mathcal{E}_x$. If $e(t)\not\in\mathcal{E}_e$, Lemma \ref{ErrorInvarianceLemma2} states that if $\exists\alpha_3\geq0$ such that \eqref{ErrorLMI2} is satisfied with $\gamma(\bar{\mathcal{L}})\leq0$, then $e(t)$ will converge to $\mathcal{E}_e$. Once $e(t)\in\mathcal{E}_e$, we know from Lemma \ref{StateConvergenceLemma} that $x(t)$ will converge to $\mathcal{E}_x$.

If agent $i$ stays connected to the network when \eqref{StayConnected} is satisfied and all of agent $i$'s neighbors are connected to the network, then $V(x(t))$ will not increase enough to destabilize the system when $e(t)$ is converging to $\mathcal{E}_e$. The integrand in \eqref{StayConnected} functions as an upper bound on the left side of the inequality in \eqref{SatisfactionCriteria}, which in turn is an upper bound on $\dot{V}(x(t))$. Consequently, the condition in \eqref{StayConnected} ensures that when all agents connect to the network, they will stay connected until $V(x(t))-V(x(0))<f(t)$, guaranteeing that $x(t)$ eventually converges to $\mathcal{E}_x$, even when $e(t)\not\in\mathcal{E}_e$.
\end{proof}

Theorem \ref{TriggerTheorem} ensures that when the magnitude of the error $V_e(e(t))$ grows too large, \eqref{TriggerCondition} will be triggered and agent $i$ will connect to the network to communicate with other agents. Note that in \eqref{TriggerCondition} and \eqref{StayConnected}, $y_i(t)$ is locally available to agent $i$, but $\gamma(\bar{t}_j)$ contains some information that is not locally available to agent $i$ since $\gamma(\bar{t}_j)$ is a function of the configuration of the overall communication graph $\mathcal{G}(t)$ $\forall t\in[\bar{t}_j,\bar{t}_{j+1})$. Let $\tilde{\mathcal{L}}_i(t)$ and $\hat{\mathcal{L}}_i(t)$ represent the portions of $\mathcal{L}(t)$ whose values are known and unknown by agent $i$, respectively, so that $\mathcal{L}(t)=\tilde{\mathcal{L}}_i(t)+\hat{\mathcal{L}}_i(t)$ $\forall i\in\{1,\cdots,N\}$. For agent $i$ to evaluate \eqref{TriggerCondition} and \eqref{StayConnected}, it first solves for $\gamma(\bar{t}_j)$ offline for all possible configurations of the communication graph $\mathcal{G}(t)$. Agent $i$ then plugs in $\bar{\gamma}_i(\bar{t}_j)$ for $\gamma(\bar{t}_j)$ in \eqref{TriggerCondition} and \eqref{StayConnected} to evaluate them online, where $\bar{\gamma}_i(\bar{t}_j)$ is given by
\begin{equation}
\medmuskip=0.72mu
\thinmuskip=0.72mu
\thickmuskip=0.72mu
\bar{\gamma}_i(\bar{t}_j)\triangleq\max_{\hat{\mathcal{L}}_i(t)}\gamma(\bar{t}_j)~\text{s.t.}~\mathcal{L}(t)=\tilde{\mathcal{L}}_i(t)+\hat{\mathcal{L}}_i(t)~\forall t\in[\bar{t}_j,\bar{t}_{j+1}).
\end{equation}
In this way, agent $i$ always evaluates the trigger condition with values that result in the right side of the inequality in \eqref{TriggerCondition} being greater than or equal to its actual value. This then functions as an upper bound on the actual value of the condition in \eqref{TriggerCondition}, which is itself an upper bound on the condition in \eqref{SatisfactionCriteria}, implying that whenever \eqref{SatisfactionCriteria} is not satisfied, the condition in \eqref{TriggerCondition} evaluated with $\bar{\gamma}_i(\bar{t}_j)$ will be triggered. Similarly, agent $i$ always evaluates \eqref{StayConnected} with values that result in the left side of the inequality in \eqref{StayConnected} being greater than or equal to its actual value, ensuring that whenever \eqref{StayConnected} is satisfied, the condition in \eqref{StayConnected} evaluated with $\bar{\gamma}_i(\bar{t}_j)$ will also be satisfied.

\begin{remark}
Note that computing $\gamma(\bar{t}_k)$ in \eqref{GammaDefinition} for all possible values of $\mathcal{L}(t)$ requires evaluating $2^N-N$ LMIs since this is the number of possible configurations of the communication graph $\mathcal{G}(t)$ for different sets of agents connected and disconnected from the network. However, all of this computation is completed offline ahead of time and grows with the number of agents $N$, not the number of control inputs $p$ or sensor measurements $m$. Furthermore, \eqref{GammaDefinition} does not need to be evaluated for each $\bar{t}_k$ since the set of possible values for $\mathcal{L}(t)$ is dependent on the configuration of the underlying communication graph $\bar{\mathcal{G}}$ and is therefore time-invariant.
\end{remark}

Theorem \ref{LMITheorem} addresses the case where the offline calculation in \eqref{GammaDefinition} becomes computationally intractable with large $N$. It states that $\gamma(\bar{t}_k)$ will always be the largest when all agents are disconnected from the network, or in other words, $V_e(e(t))$ will grow the fastest when all agents are disconnected from the network. Consequently, each agent can always use this worst-case value for $\gamma(\bar{t}_k)$ when evaluating \eqref{TriggerCondition} and \eqref{StayConnected}, and therefore only $2$ LMIs need to be evaluated in \eqref{GammaDefinition} (the cases where all agents are either connected or disconnected from the network) instead of $2^N-N$ LMIs which are needed in \cite{griffioen2021reducing}.
\begin{theorem}
\label{LMITheorem}
Let $\bar{V}_e(e(t))$ represent the value of $V_e(e(t))$ when all agents are disconnected from the network. Then
\begin{equation}
\dot{V}_e(e(t)) \leq \dot{\bar{V}}_e(e(t)),
\end{equation}
implying that $V_e(e(t))$ increases the most when all agents are disconnected from the network.
\end{theorem}
\begin{proof}
Given the error dynamics for the overall system in \eqref{ErrorDynamics},
\begin{equation}
\begin{split}
\dot{V}_e(e(t)) &= -\eta e(t)^T\left((\mathcal{L}(t)\otimes I_n)^T\bar{P}+\bar{P}(\mathcal{L}(t)\otimes I_n)\right)e(t)\\
&\quad+2e(t)^T\bar{P}\left(F(\mathcal{L}(t))e(t)+\mathcal{I}w(t)-J(\mathcal{L}(t))v(t)\right) \\
&\leq 2e(t)^T\bar{P}\left(F(\mathcal{L}(t))e(t)+\mathcal{I}w(t)-J(\mathcal{L}(t))v(t)\right) \\
&= \dot{\bar{V}}_e(e(t))
\end{split}
\end{equation}
where the inequality follows from the fact that $\eta\geq0$ and that $(\mathcal{L}(t)\otimes I_n)^T\bar{P}+\bar{P}(\mathcal{L}(t)\otimes I_n)$ is positive semidefinite since the eigenvalues of the Laplacian $\mathcal{L}(t)$ are always nonnegative.
\end{proof}

\subsection{Network Connection Procedure}
Algorithm \ref{NetworkConnectionAlgorithm} describes a procedure which ensures that $x(t)$ converges to $\mathcal{E}_x$ as guaranteed by Theorems \ref{QuadraticBoundednessTheorem} and \ref{TriggerTheorem} when \eqref{StateLMI}, \eqref{ErrorLMI2}, and \eqref{LMICondition} are satisfied. The event-triggered communication procedure presented in Algorithm \ref{NetworkConnectionAlgorithm} uses only local information to indicate when agent $i$ needs to connect to the network and communicate with other agents.
\begin{algorithm}[h!]
\small
\caption{Network Connection Procedure for Agent $i$}
\begin{algorithmic}[1]
\NoDo
\NoThen
\State Initialize $\hat{x}_i(0)$ $\forall i\in\{1,\cdots,N\}$ so that $e(0)\in\mathcal{E}_e$
\While {$1$}
\If {\eqref{TriggerCondition} is satisfied}
\State Open network connection
\State Broadcast $\hat{x}_i(t)$ and $\tau_i$ to agents $j\in\mathcal{N}_i$
\State Receive $\tau_j$ from agents $j\in\mathcal{N}_i$ connected to the network
\ParFor {Agents $j\in\mathcal{N}_i$ connected to the network}
\State Send $\{\tilde{\mathcal{L}}_i(\bar{t}_k),\bar{t}_k\}$ $\forall\bar{t}_k\in(\tau_j,t]$ to agent $j$
\State Update $\tilde{\mathcal{L}}_i(\bar{t}_k)$ with information received from $\tilde{\mathcal{L}}_j(\bar{t}_k)$
\Statex \hspace{1.5cm}$\forall\bar{t}_k\in(\tau_i,t]$
\EndParFor
\State Update $\tau_i$
\If {Some agent $j\in\mathcal{N}_i$ is not connected to the network or}
\Statex \hspace{1.3cm}\eqref{StayConnected} is not satisfied
\State Close network connection
\EndIf
\Else
\State Update $\tilde{\mathcal{L}}_i(t)$ with $a_{ij}(t)=a_{ji}(t)=0$
\Statex \hspace{1cm}$\forall j\in\{1,\cdots,N\}$
\EndIf
\State $\dot{\hat{x}}_i(t) = A_{bk}\hat{x}_i(t) + \bar{L}_i(l_{ii}(t))(y_i(t)-C_i\hat{x}_i(t))$
\Statex \hspace{1.6cm}$+\eta\sum_{j=1}^Na_{ij}(t)(\hat{x}_j(t)-\hat{x}_i(t))$
\EndWhile
\end{algorithmic}
\label{NetworkConnectionAlgorithm}
\end{algorithm}
In line 3, agent $i$ uses its local sensor measurements $y_i(t)$ as well as $\bar{\gamma}_i(\bar{t}_k)$ $\forall\bar{t}_k\leq t$ to evaluate the trigger condition in \eqref{TriggerCondition}. If \eqref{TriggerCondition} is satisfied, agent $i$ connects to the network (line 4) and broadcasts $\hat{x}_i(t)$ to its neighboring agents (line 5). Agent $i$ also broadcasts $\tau_i$ to its neighboring agents, where $\tau_i$ represents the most recent time instant where agent $i$ possesses full information about changes in the configuration of the communication graph. After receiving $\tau_j$ from each neighboring agent $j$ currently connected to the network (line 6), agent $i$ sends the information it possesses about communication graph configuration changes from time $\tau_j$ to the current time to each of these agents (line 8). Agent $i$ then uses the information it receives from these agents to update its information about communication graph configuration changes (line 9) as well as to update $\tau_i$ (line 11). If all of agent $i$'s neighbors are connected to the network, agent $i$ uses its local sensor measurements $y_i(t)$ as well as $\bar{\gamma}_i(\bar{t}_k)$ $\forall\bar{t}_k\leq t$ to evaluate the condition in \eqref{StayConnected}. If \eqref{StayConnected} is satisfied, then agent $i$ stays connected to the network and continues to communicate with its neighboring agents. Otherwise agent $i$ disconnects from the network (line 13).

If \eqref{TriggerCondition} is not satisfied, then agent $i$ updates its information about communication graph configuration changes with the fact that it is not connected to the network at time $t$ (line 16). Any state estimates agent $i$ receives from its neighboring agents are used to compute its state estimate $\hat{x}_i(t)$ according to \eqref{StateEstimateDynamics} (line 18). Note that because \eqref{TriggerCondition} functions as an upper bound on the condition in \eqref{SatisfactionCriteria}, agents will connect to the network more often than is necessary.

\begin{remark}
Note that Algorithm \ref{NetworkConnectionAlgorithm} presents a procedure where an agent's sending and receiving capabilities are simultaneously triggered by the condition in \eqref{TriggerCondition}. However, an attack on an agent is initiated through data that is incoming to that agent, not outgoing from that agent. Consequently, data could constantly be broadcast to agents all the time, while \eqref{TriggerCondition} would only be used for deciding when to receive information from other agents. In this case, agent $i$ would continuously broadcast $\hat{x}_i(t)$, $\tau_i$, and $\{\tilde{\mathcal{L}}_i(\bar{t}_k),\bar{t}_k\}$ $\forall\bar{t}_k\in(\min_{j\in\mathcal{N}_i}\tau_j,t]$ to all its neighbors. By doing so, an agent receiving information would possess the state estimates from all its neighbors, reducing that agent's state estimation error compared to the current scenario where only a subset of the neighbors' state estimates may be received. This in turn would decrease the number of times \eqref{TriggerCondition} is triggered since \eqref{TriggerCondition} is a function of the estimation error, further reducing an adversary's window of opportunity to carry out an attack. However, this approach would increase communication costs considerably since all agents would always be broadcasting information. The implementation of this approach, along with an investigation of the tradeoff between overall performance and communication costs, is left for future work.
\end{remark}

\begin{remark}
For the discrete time case, \eqref{StateLMI}, \eqref{ErrorLMI}, \eqref{ErrorLMI2}, \eqref{LMICondition}, \eqref{TriggerCondition}, and \eqref{StayConnected} would be replaced by
\begin{equation}
\medmuskip=3.78mu
\thinmuskip=3.78mu
\thickmuskip=3.78mu
\begin{bmatrix}
(1-2\alpha_1)P-A_{bk}^TPA_{bk} & A_{bk}^TPE & -A_{bk}^TP \\
E^TPA_{bk} & \alpha_1\bar{P}-E^TPE & E^TP \\
-PA_{bk} & PE & \alpha_1Q-P
\end{bmatrix}\succ0,
\end{equation}
\begin{equation}
\tiny
\medmuskip=-2.59mu
\thinmuskip=-2.59mu
\thickmuskip=-2.59mu
\begin{bmatrix}
(\gamma(\mathcal{L}_k)-2\alpha_2)\bar{P}-A_e(\mathcal{L}_k)^T\bar{P}A_e(\mathcal{L}_k) & -A_e(\mathcal{L}_k)^T\bar{P}\mathcal{I} & A_e(\mathcal{L}_k)^T\bar{P}J(\mathcal{L}_k) \\
-\mathcal{I}^T\bar{P}A_e(\mathcal{L}_k) & \alpha_2Q-\mathcal{I}^T\bar{P}\mathcal{I} & \mathcal{I}^T\bar{P}J(\mathcal{L}_k) \\
J(\mathcal{L}_k)^T\bar{P}A_e(\mathcal{L}_k) & J(\mathcal{L}_k)^T\bar{P}\mathcal{I} & \alpha_2R-J(\mathcal{L}_k)^T\bar{P}J(\mathcal{L}_k)
\end{bmatrix}\succ0,
\end{equation}
\begin{equation}
\tiny
\medmuskip=-2.19mu
\thinmuskip=-2.19mu
\thickmuskip=-2.19mu
\label{ErrorLMI2Discrete}
\begin{bmatrix}
(\gamma(\bar{\mathcal{L}})-2\alpha_3)\bar{E}^T\bar{P}\bar{E}-\bar{A}_e^T\bar{E}^T\bar{P}\bar{E}\bar{A}_e & -\bar{A}_e^T\bar{E}^T\bar{P}\bar{E}\bar{W} & \bar{A}_e^T\bar{E}^T\bar{P}\bar{E}\bar{V} \\
-\bar{W}^T\bar{E}^T\bar{P}\bar{E}\bar{A}_e & \alpha_3Q-\bar{W}^T\bar{E}^T\bar{P}\bar{E}\bar{W} & \bar{W}^T\bar{E}^T\bar{P}\bar{E}\bar{V} \\
\bar{V}^T\bar{E}^T\bar{P}\bar{E}\bar{A}_e & \bar{V}^T\bar{E}^T\bar{P}\bar{E}\bar{W} & \alpha_3R-\bar{V}^T\bar{E}^T\bar{P}\bar{E}\bar{V}
\end{bmatrix}\succ0,
\end{equation}
\begin{equation}
\footnotesize
\medmuskip=1.28mu
\thinmuskip=1.28mu
\thickmuskip=1.28mu
\begin{bmatrix}
P-A_{bk}^TPA_{bk}-C_i^TY_iC_i & A_{bk}^TPE & -A_{bk}^TP & -C_i^TY_i\Gamma_i \\
E^TPA_{bk} & \bar{P}-E^TPE & E^TP & 0 \\
-PA_{bk} & PE & Q-P & 0 \\
-\Gamma_i^TY_iC_i & 0 & 0 & R-\Gamma_i^TY_i\Gamma_i
\end{bmatrix} \succeq 0,
\end{equation}
\begin{equation}
y_k^{i^T}Y_iy_k^i \leq 2 + \max\left(1,\prod_{j=0}^{k-1}\gamma(\mathcal{L}_j)\right),
\end{equation}
\begin{equation}
\sum_{\ell=0}^k 2 + \max\left(1,\prod_{j=0}^{\ell-1}\gamma(\mathcal{L}_j)\right) - y_\ell^{i^T}Y_iy_\ell^i > f(k+1),
\end{equation}
respectively, where $\mathcal{L}_k$ represents the Laplacian at time step $k$ and $y_k^i$ denotes agent $i$'s sensor measurements at time step $k$. However, \eqref{ErrorLMI2Discrete} may not be satisfied with $\gamma(\bar{\mathcal{L}})\in[0,1]$ if it is not possible to design the coupling gain $\eta$ so that the eigenvalues of the lower right block of $\bar{A}_e$ lie within the unit circle.
\end{remark}

\subsection{Ensuring Resiliency Against Attacks}
The network connection procedure presented in Algorithm \ref{NetworkConnectionAlgorithm} is sufficient for ensuring that agents connect to the network when necessary to maintain the stability of the overall system in attack-free scenarios. However, during those brief periods of time when various agents are connected to the network, the safety of the overall system against attacks is not guaranteed. Since resilience against attacks is the ultimate goal, a variety of mechanisms and strategies may be used during these brief periods of network connection to guarantee safety and security. For example, software rejuvenation \cite{griffioen2019secure,griffioen2020secure} is one mechanism that has been introduced to guarantee the safety of agents when connecting to the network to maintain stability or recover from a disturbance. The detailed implementation of such a mechanism within the context of the network connection protocol in Algorithm \ref{NetworkConnectionAlgorithm} is beyond the scope of this article and is left for future work. However, to guarantee the safety of the overall system in the presence of attacks, some such resiliency mechanism will need to be implemented during those brief periods of time when various agents connect to the network and share critical information.

\section{Simulation}
To illustrate the effectiveness of the network connection and communication protocol, we consider a smart water distribution system used at a four-hectare wine estate in the south of England \cite{fu2020dynamic}. The goal of the water distribution system is to stabilize the water levels of three district meter area tanks at predesigned constant reference levels. The system state is given by the difference between the reference levels and the current water levels of the three tanks, the control inputs are the open levels of the valves, and the sensors measure the current water levels of the tanks. The global system model is linearized at a reference level of $3$ m as presented in \cite{fu2020dynamic} and is given by
\begin{align}
\dot{x}(t) &=
\begin{bmatrix}
-8.367&0&0\\
0&-6.276&0\\
0&0&-5.020
\end{bmatrix}\times10^{-4}x(t) \\
&\quad+\begin{bmatrix}
0.1068&-0.0371&-0.0371\\
-0.0279&0.0801&-0.0279\\
-0.0223&-0.0223&0.0641
\end{bmatrix}u(t)+w(t), \nonumber\\
y(t) &=
\begin{bmatrix}
1&0&0\\
0&1&0\\
0&0&1
\end{bmatrix}x(t)+v(t).
\end{align}
We let $Q=\frac{10000}{3}I_n$ and $R=\frac{10000}{3}I_m$ so that the process and measurement disturbances are less than $1$ cm for each tank.

A continuous time controller $K$ with three poles at $-1.5$ is designed to stabilize the system when all agents are connected to the network, and continuous time observers $L$ and $\hat{L}_i$ $\forall i\in\{1,\cdots,N\}$ are designed according to \eqref{ObserverDesign} so that the estimation error for the observable states is stabilized. $L$ is designed with three poles at $-100$, $L_i^o$ is designed with a pole at $-15$ $\forall i\in\{1,\cdots,N\}$, and the coupling gain is chosen as $\eta=100000$. The water distribution system is comprised of $N=3$ agents, where each agent has access to one local control input and one local sensor measurement, and the adjacency matrix for the underlying communication graph is given by
\begin{equation}
\bar{\mathcal{A}} =
\begin{bmatrix}
0&1&0\\
1&0&1\\
0&1&0
\end{bmatrix}.
\end{equation}
We solve for $P$, $\bar{P}$, and $Y_i$ $\forall i\in\{1,\cdots,N\}$ so that \eqref{StateLMI}, \eqref{ErrorLMI2}, and \eqref{LMICondition} are satisfied with $\gamma(\bar{\mathcal{L}})\leq0$. The network connection protocol in Algorithm \ref{NetworkConnectionAlgorithm} is executed for each agent from the initial state $x(0)=\begin{bmatrix}10&10&10\end{bmatrix}^T$. To model changes in the setpoint as well as impulsive disturbances, we periodically update the state so that it jumps to a new value, simulating the behavior of a real-world system.

Figures \ref{fig:LyapunovFunction} and \ref{fig:NetworkConnections} depict the Lyapunov function convergence and network connection timeline, respectively, for a particular simulation. As seen in Figure \ref{fig:LyapunovFunction}, the Lyapunov function continually decreases between setpoints, which each cause a jump in the value of the Lyapunov function. Figure \ref{fig:NetworkConnections} depicts the detailed connection timeline for each agent, showing that agents are able to disconnect from the network approximately $49\%$ of the time. This provides less time for adversaries to attack different agents while also ensuring the stability of the overall system when there is no attack.
\begin{figure}[h!]
\centering
\includegraphics[width=\columnwidth]{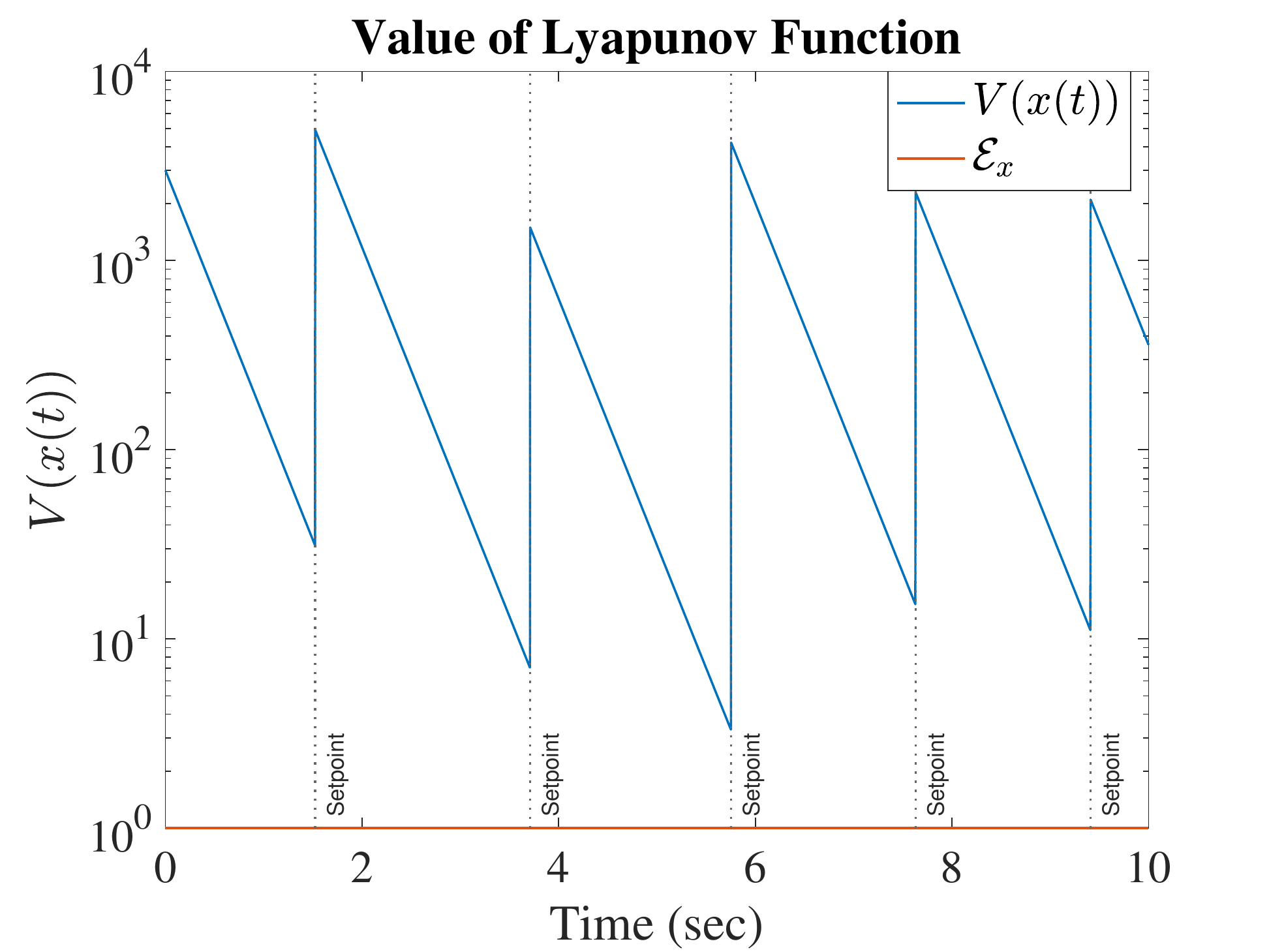}
\caption{Convergence of the Lyapunov function between setpoints.}
\label{fig:LyapunovFunction}
\end{figure}
\begin{figure}[h!]
\includegraphics[width=\columnwidth]{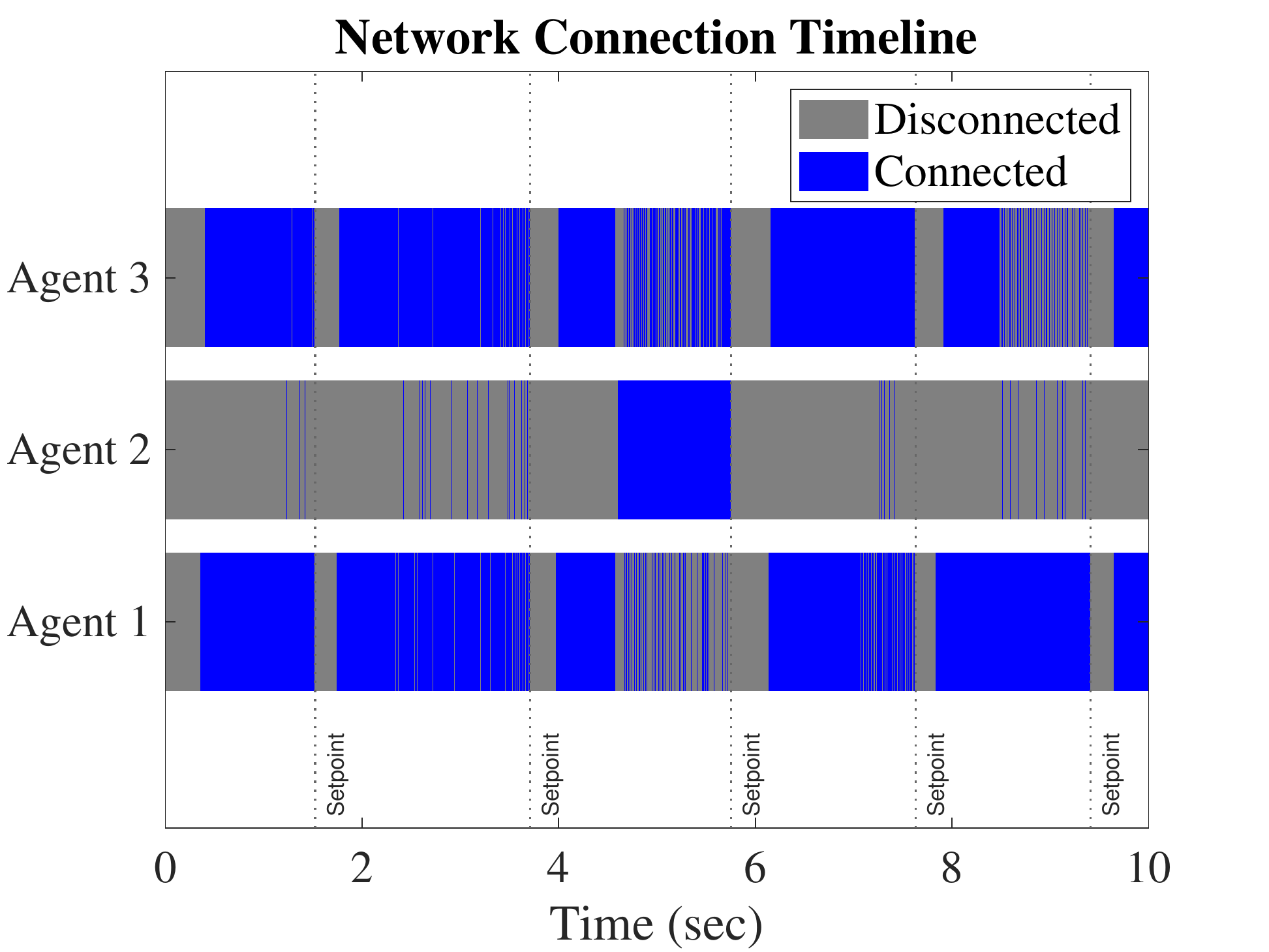}
\caption{Timeline of agents' network connections (agents remain disconnected approximately $49\%$ of the time).}
\label{fig:NetworkConnections}
\end{figure}

Note that the network connection times may vary for each agent since each agent has different sets of local sensors, since the dynamics of some agents may be more tightly coupled to one another than the dynamics of other agents, and since some agents may have more neighbors than other agents. These factors cause some agents to connect to the network more than others to ensure the stability of the overall system. Since agent $2$ is the only agent that has $2$ neighbors, it receives more information each time it connects to the network, and consequently it connects to the network far less than agents $1$ and $3$, as seen in Figure \ref{fig:NetworkConnections}.

In addition, no performance is lost in using the decentralized event-triggered network connection protocol. Over $1000$ trials with no setpoint updates, the time taken to converge to the invariant set $\mathcal{E}_x$ remains the same regardless of whether communication between agents occurs all the time (average convergence time of $2.6728$ sec) or whether agents disconnect from the network for intermittent periods of time (average convergence time of $2.6749$ sec). This average convergence time is significantly less than the average convergence time of $19.982$ sec in \cite{griffioen2021reducing} due to a more efficient state estimation scheme.

\section{Conclusion}
This article has investigated using decentralized event-triggered control to reduce attack opportunities. An event-triggered mechanism for network connection and communication is designed based on only local information. This mechanism ensures the stability of the overall system for attack-free scenarios. It also allows agents to disconnect from the network for periods of time, minimizing an adversary's window of opportunity when attacking different agents. A network connection protocol is designed which uses this event-triggered mechanism, and its effectiveness is illustrated in the context of a smart water distribution system. To ensure safety and security against attacks, future work should introduce resiliency mechanisms for those times when agents are connected to the network and are vulnerable to attacks. Future work also includes considering scenarios where non-negligible communication delays exist when sending data over the network.

\bibliographystyle{IEEEtran}
\bibliography{root}

\begin{IEEEbiography}[{\includegraphics[width=1in,height=1.25in,clip,keepaspectratio]{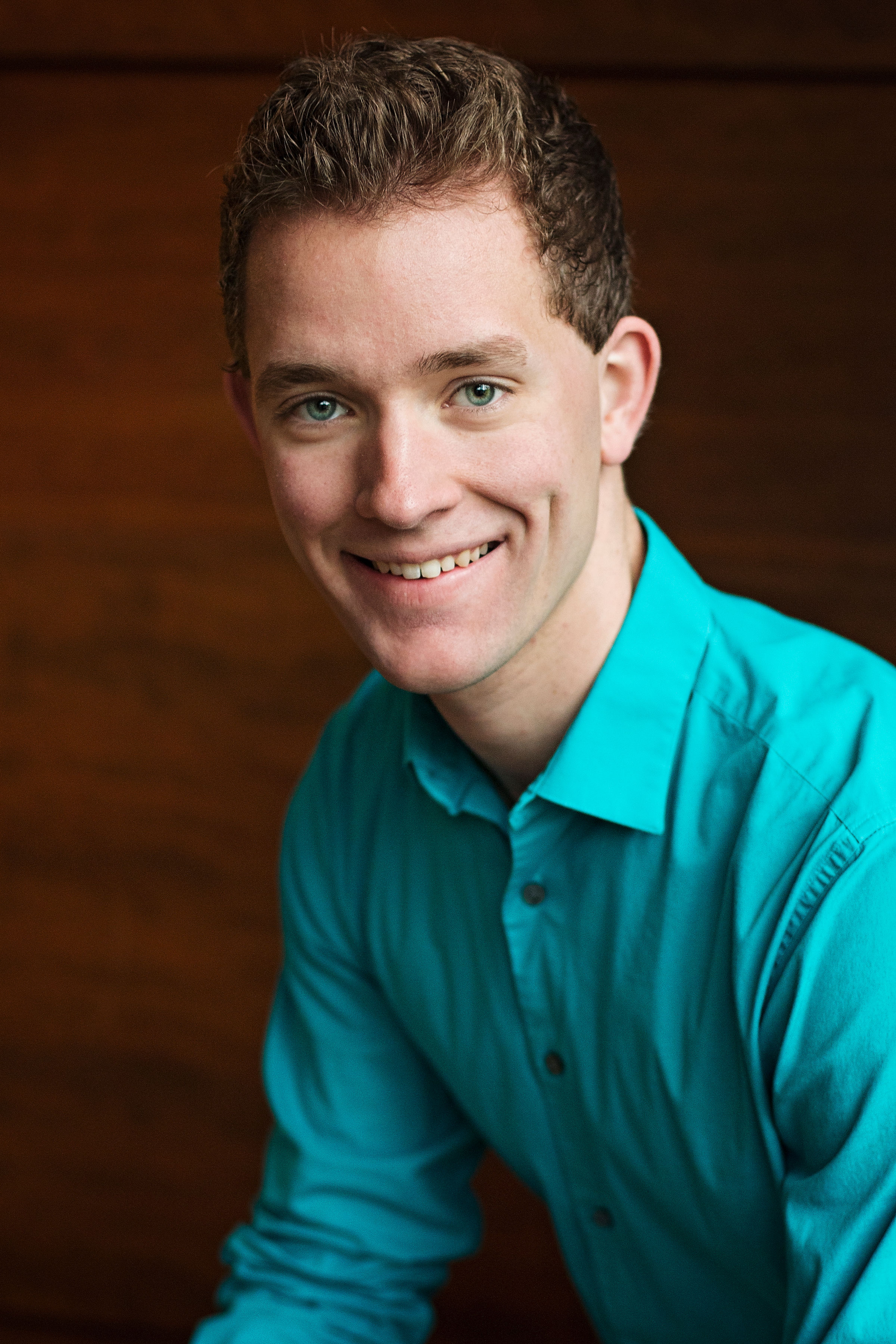}}]{Paul Griffioen}
received the B.S. degree in Engineering, Electrical/Computer concentration, from Calvin College, Grand Rapids, MI, USA in 2016 and the M.S. degree in Electrical and Computer Engineering from Carnegie Mellon University, Pittsburgh, PA, USA in 2018. He is currently pursuing the Ph.D. degree in Electrical and Computer Engineering at Carnegie Mellon University. His research interests include the modeling, analysis, and design of active detection techniques and response mechanisms for ensuring resilient and secure cyber-physical systems.
\end{IEEEbiography}
\begin{IEEEbiography}[{\includegraphics[width=1in,height=1.25in,clip,keepaspectratio]{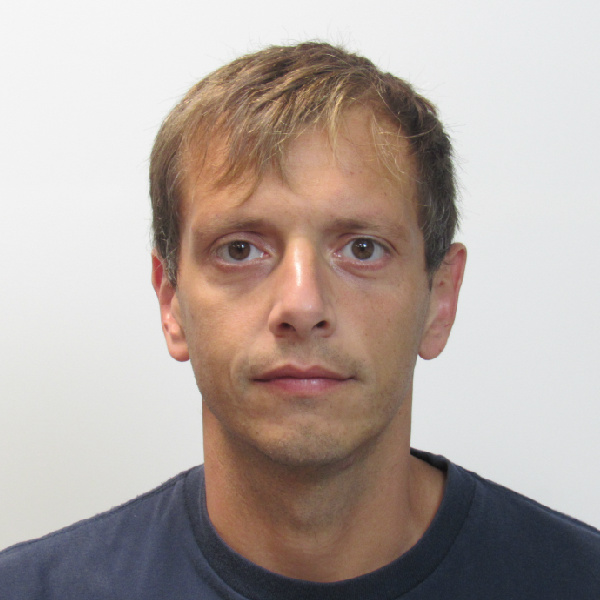}}]{Raffaele Romagnoli}
received the Ph.D. degree in control system and automation specialized in optimal and robust control system theory from the Universit\'{a} Politecnica delle Marche (UNIVPM), Ancona, Italy, in 2015. From 2015 to 2017, he was a Postdoctoral Researcher with the Department of Control Engineering and System Analysis (SAAS), Universit\'{e} Libre de Bruxelles (ULB), Brussels, Belgium. After being a Postdoctoral Researcher with the Department of Electrical and Computer Engineering, Carnegie Mellon University (CMU), Pittsburgh, PA, USA, he is now a Research Scientist in the same institution working on safe and secure control of AI robotics applications and edge computing. He is still currently collaborating with the Software Engineering Institute at CMU where he has been working on secure control for cyber-physical systems.  His other research interests are nonlinear control, control of  biological systems, energy storage (Li-ion batteries), space applications in microgravity conditions, and model inversion.
\end{IEEEbiography}
\begin{IEEEbiography}[{\includegraphics[width=1in,height=1.25in,clip,keepaspectratio]{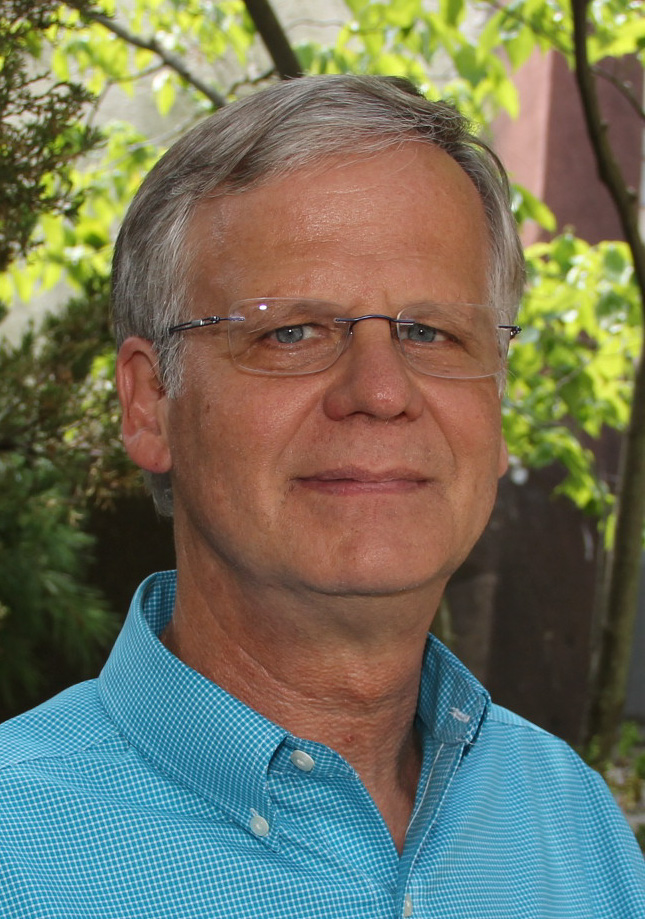}}]{Bruce H. Krogh} is professor emeritus of electrical and computer engineering at Carnegie Mellon University in Pittsburgh, PA, USA, and a member of the technical staff of Carnegie Mellon's Software Research Institute. He was founding director of Carnegie Mellon University-Africa in Kigali, Rwanda. He is chair of the board of the Kigali Collaborative Research Centre (KCRC) in Rwanda and co-lead of the IEEE Continu$\blacktriangleright$ED initiative to develop IEEE's continuing education resources for technical professionals in Africa. Professor Krogh’s research is on the theory and application of control systems, with a current focus on methods for guaranteeing safety and security of cyber-physical systems. He was founding Editor-in-Chief of the \textit{IEEE Transactions on Control Systems Technology}. He is a Life Fellow of the IEEE and a Distinguished Member of the IEEE Control Systems Society.
\end{IEEEbiography}
\begin{IEEEbiography}[{\includegraphics[width=1in,height=1.25in,clip,keepaspectratio]{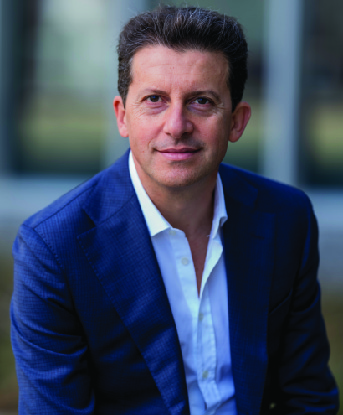}}]{Bruno Sinopoli} is the Das Family Distinguished Professor at Washington University in St. Louis, where he is also the founding director of the center for Trustworthy AI in Cyber-Physical Systems and chair of the Electrical and Systems Engineering Department. He received the Dr. Eng. degree from the University of Padova in 1998 and his M.S. and Ph.D. in Electrical Engineering from the University of California at Berkeley, in 2003 and 2005 respectively. After a postdoctoral position at Stanford University, Dr. Sinopoli was member of the faculty at Carnegie Mellon University from 2007 to 2019, where he was a professor in the Department of Electrical and Computer Engineering with courtesy appointments in Mechanical Engineering and in the Robotics Institute and co-director of the Smart Infrastructure Institute. His research interests include modeling, analysis and design of Resilient Cyber-Physical Systems with applications to Smart Interdependent Infrastructures Systems, such as Energy and Transportation, Internet of Things and control of computing systems.
\end{IEEEbiography}

\end{document}